\documentclass[12pt]{article}
\usepackage{ae} % or {zefonts}
\usepackage[T1]{fontenc}
\usepackage{amssymb}
\usepackage{latexsym}
\usepackage{amsmath}
\usepackage[ansinew]{inputenc}
\usepackage{euscript}
\usepackage{epsfig}
\usepackage{multirow}
\usepackage{anysize}

\newcommand{\qed}{\hspace*{\fill}
            $\Box$\smallskip}
\newcommand{\thmqed}{\hspace*{\fill}
            $\blacksquare$\smallskip}

\newenvironment{proof}{\noindent {\bf Proof:} \par}
                      {\qed}

 \newtheorem{theorem}{Theorem}[section]
 \newtheorem{corollary}{Corollary}[section]
 \newtheorem{lemma}{Lemma}[section]
 \newtheorem{example}{Example}[section]

 \newtheorem{proposition}{Proposition}[section]
\begin{document}
\title{Greediness and Equilibrium in Congestion Games}
\author{Sergey Kuniavsky \footnote{Corresponding
author: Munich Graduate School of Economics, Kaulbachstr. 45, Munich 80539, Germany. Financial support from the Deutsche Forschungsgemeinschaft through GRK 801 is gratefully acknowledged. $<${\tt Sergey.Kuniavsky@lrz.uni-muenchen.de}$>$.} and Rann Smorodinsky \footnote{Faculty of Industrial Engineering and Management, Technion,
Haifa 32000, Israel. $<${\tt rann@ie.technion.ac.il $>$}}}

\maketitle

\begin{abstract}
Rosenthal (1973) introduced the class of congestion games and proved that they always possess a Nash equilibrium in pure strategies. Fotakis et al. (2005) introduce the notion of a greedy strategy tuple, where players sequentially and irrevocably choose a strategy that is a best response to the choice of strategies by former players. Whereas the former solution concept is driven by strong assumptions on the rationality of the players and the common knowledge thereof, the latter assumes very little rationality on the players' behavior. From Fotakis \cite{fotakis10} it follows that for Tree Representable congestion Games greedy behavior leads to a NE. In this paper we obtain necessary and sufficient conditions for the equivalence of these two solution concepts. Such equivalence enhances the viability of these concepts as realistic outcomes of the environment.  The conditions for such equivalence to emerge for monotone symmetric games is that the strategy set has a tree-form, or equivalently is a `extension-parallel graph'. 
\end{abstract}

\section{Introduction}
Congestion games form a natural class of games that are useful in modeling many realistic settings, such as traffic and communication networks, routing, load balancing and more.

A symmetric congestion game is a 4-tuple $(N,R,\Sigma, \{P_r\}_{r\in R})$,  where $N$ is a finite set of players,
$R$ is a finite set of resources, $\Sigma \subset 2^R$ is the set of players' strategies, and for any $r \in R$, $P_r: N \to \mathbb{R}$ is the resource's payoff function. A strategy of a player is a choice of a subset resources, $s \in \Sigma$.
For any strategy tuple $s = (s^i)_{i \in N} \in \Sigma^N$ let $c(s)_r  = |\{i \in N : r \in s^i\}| $ denote the number of players that utilize $r$ (a.k.a. the congestion of the resource $r$) and denote by $c(s) = (c(s)_1, \ldots c(s)_r)$ the congestion vector. The utility of a player is the total payment for the resources she utilizes.
Formally,  $U^i(s) = \Sigma_{r \in s^i} P_r(c(s)_r).$

A congestion game is {\em monotone} if for for any $1\leq k < l \leq N$ and $r \in R$, $P_r(k) > P_r(l)$.
Monotone congestion games widely prevail in modeling traffic and communication problems, production resource allocation and more. In {\em single-signed} congestion games payments are either all positive or all negative.
Typically, whenever monotone congestion games are used for modeling, they are assumed single-signed.

A {\em congestion game form} is a pair $F=\{R , \Sigma\}$, composed of the set of resources and a set of strategies (subsets of $R$). For any congestion game $G =(N,R,\Sigma, \{P_r\}_{r\in R})$ let $F(G)=(R,\Sigma)$ denote the corresponding game form. Given a congestion game form $F$, let ${\cal G}(F) = \{G: (F(G)=F) \wedge (G$ monotone$) \}$ denote the class of all monotone congestion games with the game form $F$.

We say that a strategy set, $\Sigma \subset R$, is {\em subset-free} if for any $s \not =t \in \Sigma$ we have $s \not \subset t$ and $s \not \subset t$. Thus, a {\em a subset-free Congestion Game (Form)} is a Congestion Game (Form) with a subset-free strategy space. For any equilibrium analysis of single-signed monotone congestion games the assumption of subset free strategy sets is without loss of generality. In particular, note that in such games for any pair of strategies $s \subset t$ in $\Sigma$ either $s$ is dominated by $t$ (in case resource payments are all positive) or $t$ is dominated by $s$ (in case resource payments are all negative) and so after deletion of dominated strategies we are left with subset free sets.

As usual, a profile $s \in \Sigma^N$ is a pure NE of $G$, if for each player $i$, for each strategy $t^i \in \Sigma$,
$U^i(s^i,s^{-i}) \geq U^i(t^i,s^{-i})$, where $s^{-i}$ is the vector of strategies of all players but $i$.
Informally, a set of strategies is a Nash equilibrium if no player can do better by unilaterally deviating. The set of all pure NE of a congestion game $G$ will be denoted $NE(G)$.

\subsection{Known Results}

Congestion games were introduced by Rosenthal (1973) \cite{Rosental} , who proved that any congestion game has a Nash equilibrium in pure strategies. In spite of this fact there is still valid concern about the prevalence of a Nash equilibrium in reality. There are two classical criticisms over the validity of a Nash equilibrium profile as a {\it solution concept} which can be made - one that is computationally driven and another that is rationality driven.

The latter criticism is based on the fact that for an equilibrium to prevail players must have common knowledge of rationality, a condition typically unrealistic. The former criticism argues that the existence of a pure Nash equilibrium  does not imply it is computationally simple to find such an equilibrium. In particular, whenever the strategy space is rich this may be a challenging endeavor. In general, this may require searching over all strategy tuples, whose number can grow exponentially with the number of players. However, as congestion games have the {\em finite improvement property} one could suspect that it may be easier to find such an equilibrium.%
 \footnote{The finite improvement property asserts that if players sequentially improve their utility by unilateral strategy changes then this process is finite and must end in a Nash equilibrium profile, see, Monderer and Shapley \cite{MondSh}.}
 However, it turns out that improvement paths can be exponentially long, as demonstrated by Ieong et al.  \cite{Ieong} and Fabricant et al. \cite{Fabrikant et al}. %
%\footnote{Ieong et al.  \cite{Ieong et al} show that such improvement paths are polynomial when the strategy space is composed of singleton resource sets.}
Recently Fabrikant et al. \cite{Fabrikant et al} provide an algorithm that finds a Nash equilibrium in polynomial time, for an important subset of congestion games. This is done via a reduction to a flow problem, yet leaves little insight regarding the nature of the NE.

%On the other hand, one could question whether an equilibrium will typically prevail in such games. Recall that the emergence of a Nash equilibrium strategy profile is generally the result of the rationality as well as the common knowledge of rationality of the players \textbf{reference}.

Fotakis et al. (2005) \cite{Fotakis05} introduce the notion of a {\em greedy strategy profile}. Let us consider a dynamic setting with the players joining the game sequentially. Each player, upon arrival, irrevocably
chooses a best response strategy to the choice of strategies of the previous players,
while ignoring subsequent players. The resulting strategy profile is called a {\em greedy strategy profile}. Let us denote by $Z(G)$ the set of all greedy strategy profiles. Note the two degrees of freedom in the process - the order of the agents and the tie-breaking rule in case of indifference among several options.

Formally, $s \in Z(G)$, if there exists a permutation $\pi:N \to N$
of the players ($\pi(i)$ denotes the order of $i$) such that for any
player $i$ who chooses strategy $s^i$ we have $\sum_{r \in s^i}P_r(c(s^i_\pi)_r+1) \geq \sum_{r \in t}P_r(c(s^i_\pi)_r+1)$
$\ \forall t \in \Sigma$,
where $c(s^i_{\pi})_r = |\{j: r \in s^{\pi(j)}, \pi(j)<\pi(i)\}|$
is the number players preceding $i$ according to the permutation $\pi$  whose strategy includes resource $r$. Clearly $Z(G) \not = \emptyset$, and typically $Z(G)$ may contain many such profiles.

In contrast with the rationality assumption underlying the notion of Nash equilibrium, the rationality requirement from a greedy profile is very low, possibly too low, as players clearly choose to ignore anything they do not observe.
In addition, calculating a greedy equilibrium profile is a much less demanding task than calculating a Nash equilibrium. Hypothetically, whenever $NE(G)=Z(G)$ the prevalence of an equilibrium is much more likely. This is because the identity of the two sets suggests that the rationality assumption underlying an equilibrium profile is weak and the complexity of finding such an equilibrium is typically linear with respect to the number of players. This motivates us to study the relationship between the sets $NE(G)$ and $Z(G)$.

Fotakis et al. \cite{Fotakis05} have already shown that  $Z(G) \subset NE(G)$ for simple congestion games. Fotakis \cite{fotakis10} showed that if a class of congestion games that satisfy two conditions: (1) the game form is that of `extension-parallel graph', namely one can map the resources to the set of edges in a  extension-parallel graph and the strategies are the set of paths leading from a certain node in the graph (designated as the source node) to another node (designated as the target node); and (2) the resource payoff functions satisfy a property referred to as the `Common Best Reply' requirement, met in symmetric congestion game. In particular, their result implies that  $Z(G) \subset NE(G)$ for simple congestion games, where strategies are the singleton resource subsets of $R$.

%A `tree representable game form', introduced in Holzman and Law-Yone \cite{Holz}, is a game form where the resources  can be mapped to the nodes of a tree graph and the strategies are all paths leading from the root node to some leaf node. Holzman and Law-Yone show that all Nash equilibria in such games must be strong equilibria.\footnote{A strong equilibrium is a Nash equilibrium for which no coalition has a profitable deviation.}Our main result states that a `tree representable game form' is a necessary and sufficient condition on the game form to guarantee that equivalence of $NE(G)$ and $Z(G)$ for monotone games.One corollary of our main result is the following observation about simple games, namely games where the strategy set $\Sigma$ is the set of singleton resources: \begin{corollary}\label{thm for simple games} If $F=(R,\Sigma)$ and $\Sigma =\{\{r\}: r \in R\}$ then for any $G\in {\cal G}(F)$, $NE(G)=Z(G)$.\end{corollary} As previously stated, the fact that $Z(G) \subset NE(G)$ for simple congestion games, follows from the results of Fotakis et al. \cite{Fotakis05}. The literature on various settings of congestion games and NE is a wide one, and includes many papers. Many of the important subsets of congestion games were introduced by Milchteich in \cite{milchteich} and Holzman and Law Yone in \cite{holz03} showed an equivalence between two important sets of congestion games, namely, 'extension parallel networks' and 'tree representable congestion games'.

Additional papers that study conditions under which $Z(G) \subset NE(G)$ are Ackerman et al. \cite{Ackerman et al} and Fotakis \cite{fotakis10}. In Ackerman et al. \cite{Ackerman et al} the main observation is that greedy best responses converge very fast to a NE, when the strategy structure is that of a Matroid, while in
Fotakis \cite{fotakis10} shows a theorem from which follows that in Tree Representable Congestion Games greedy leads to NE.

\subsection{Our Contribution}

This paper characterizes the setting for which $Z(G)$ and $NE(G)$ coincide. In particular, our main result argues that a necessary and sufficient conditions for these two solution concepts to coincide is that the the game form is that of `extension-parallel graph'. These results extend the state of the art knowledge in two ways. First, it is shown that for such game forms not only is every greedy profile a Nash equilibrium but also vice versa. In addition, we show that for such equivalence to hold for a given game form it must be the4 case that the game form is of a certain class, namely a `extension-parallel graph'. In particular, given a game form not satisfying this condition, we show how to construct  resource payoff functions such that the the set of NE profiles and greedy profiles will not coincide.

The `extension-parallel graph' game form is also the necessary and sufficient condition for the set of NE profiles to coincide with the set of strong equilibrium profiles, as shown by Holzman and Law Yone \cite{Holz} and \cite{holz03}. Note that Holzman and Law Yone \cite{holz03} refer `tree representable' game form which, a-priori, are different than `extension-parallel graph' game forms, but they go on and prove equivalence (Theorem 1 and Theorem 2).

%SERGEY - CAN WE GET THE HOLZMAN LAW-YONE RESULT AS  A COROLLARY OF OURS?
Combining the results in \cite{Holz} and \cite{holz03} and our contribution, we obtain equivalence for 'extension-parallel graph' game forms (or Tree Representable game forms) between greedy profiles and strong NE. Moreover, if the game is not form is not 'extension parallel graph' it is possible to find payoffs where the equivalence will not hold.

The structure of the article is as follows: Section 2 provides a variety of examples that demonstrate that without any restrictions on the game form there is no connection between the sets $NE(G)$ and $Z(G)$. Section 3 formalizes the notion of tree representable games, and discusses the characteristics of this class. Then in section 4 we present and prove the main result, namely equivalence between $NE(G)$ and $Z(G)$ for tree representable congestion games.

\section{Examples}

Here we provide several examples for the various relations between $Z(G)$ and $NE(G)$. As we shall demonstrate those can differ depending on the game in question.

\begin{example}\label{3choose2}
$NE(G) \cap Z(G) = \emptyset$ - Greedy profiles and equilibria are mutually exclusive.

In this example there are 3 players and 3 resources. The strategy space is the set of all pairs of resources.

\begin{tabular}{|c|c|c|c|}
\hline
\# of players / Resource & A & B
& C \\
\hline
1 & 10 & 10 & 8 \\
\hline
2 & 8 & 4 & 6 \\
\hline
3 & 1 & 1 & 5 \\
\hline
\end{tabular}\\

The unique greedy profile (up to renaming of players) is $(AB,AC, BC)$. Note this is not a Nash equilibrium since player 1 can profitably deviate from AB to AC, increasing her utility from 8+4 to 8+5. On the other hand the unique Nash equilibrium (up to renaming of players) is  $AC, AC, BC$, and obtained after this deviation.
\end{example}

\begin{example}\label{2layer}
$Z(G) \subsetneq NE(G)$  - Greedy profiles strictly contained in NE profiles.

In this example there are 2 players and 4 resource. Each strategy must contain one of the resources A,B and one of the resources C,D.

\begin{tabular}{|c|c|c|c|c|}
\hline
\# of players / Resource & A & B & C & D\\
\hline
1  & 40 & 30 & 20 & 15 \\
\hline
2  & 10 & 11 & 12 & 13 \\
\hline
\end{tabular}\\

Clearly the unique greedy profile is $(AC,BD)$, which is also a Nash equilibrium. However, there is an additional Nash equilibrium profile $(AD,BC)$.
\end{example}

\begin{example}\label{GreedySuperset}
$NE(G) \subsetneq Z(G)$ - The greedy profiles strictly contain the NE.

In this example there are 2 players and 3 resources. The strategy space is the set of all pairs of resources.

\begin{tabular}{|c|c|c|c|}
\hline
\# of players / Resource & A & B & C \\
\hline
1 & 10 & 8 & 8 \\
\hline
2 & 1 & 7 & 6 \\
\hline
\end{tabular}\\

The greedy profiles are $AB, BC$ and $AC, BC$. The first one is the unique pure NE of the game.
\end{example}

\begin{example}
$Z(G) \cap  NE(G) \not = \emptyset$, $Z(G) \setminus  NE(G)  \not = \emptyset$ and $NE(G) \setminus  Z(G)  \not = \emptyset$.

Consider a game with 2 players, 5 resources: A,B,C,D,E and the strategy space $\Sigma=\{ AB, AC, DB, E\}$.

\begin{tabular}{|c|c|c|c|c|c|}
\hline
\# of players / Resource & A & B & C & D & E\\
\hline
1 & -1 & -1 & -5 & -2 & -10\\
\hline
2 & -5 & -10 & -100 & -100 & -100 \\
\hline
\end{tabular}\\

If we express this game in the standard bi-matrix form we get:

\begin{tabular}{|c|c|c|c|c|}
\hline
Game & AB & AC & DB & E \\
\hline
AB & -15, -15 & -6, \textbf{-10} & -11, -12 & \textbf{-2}, \textbf{-10} \\
\hline
AC & \textbf{-10}, -6 & -105, -105 & \textbf{-6}, \textbf{-3} & -6, -10 \\
\hline
DB & -12, -11 & \textbf{-3}, \textbf{-6} & -110, -110 & -3, -10 \\
\hline
E & \textbf{-10}, \textbf{-2} & -10, -6 & -10, -3 & -100, -100 \\
\hline
\end{tabular}\\

It is now easy to verify that the only two Nash equilibria of this game are $(AB,E)$ and $(AC,DB)$. The greedy profiles, on the other hand, are  $(AB, E)$ and $(AB,AC)$.
\end{example}

\begin{example}
$NE(G) = Z(G)$ - Equivalence

This holds in any simple congestion game (follows from our main result).
\end{example}

\section{Tree Representable Congestion Games}

Holzman and Law Yone (1996) \cite{Holz} study a class of game forms called \textit{Tree Representable Congestion Games} (TRCG).
To define this they introduce the notion of an {\em $R$-tree}. An $R$-tree is a tree whose nodes,
except for the root, are labeled by elements in $R$, each appearing at most once
(however not necessarily all elements in $R$ are mapped to the nodes). With each terminal node in the tree we can associate the set of resources that form the unique path leading from the root to that node. Thus, the $R$-tree induces a set of strategies.

The game form $(R,\Sigma)$ is {\em Tree Representable} if there exists an $R$-tree which induces the set $\Sigma$. A congestion game $G$ is tree representable iff its corresponding game form is tree representable.%
\footnote{
Actually, Holzman and Law Yone (1996) \cite{Holz} allow for strategies that are induced by paths that do not necessarily lead to terminal nodes. However, whenever the game is subset-free this cannot occur and hence our focus is only on terminal nodes. In \cite{holz03} they introduce the notion of subset free and adjust the tree representable definition.
}

\begin{example}
Consider a game form $(R,\Sigma)$ with $R=\{ A, B, C, D, E, F, G, H, I, J, K, L \}$ and
$\Sigma=\{ ABG; AH; CI; CFJ; DEK; DEL \}$. This game form is tree representable and the corresponding tree is depicted in Figure \ref{RTree}, where the set of terminal nodes is $\{ G, H, I, J, K, L \}$. Note that the tree in the figure is not unique. For example, we can exchange nodes $F$ and $J$ and still represent the same game.
\end{example}

\begin{figure}
 \centering
    \includegraphics[height=80mm]{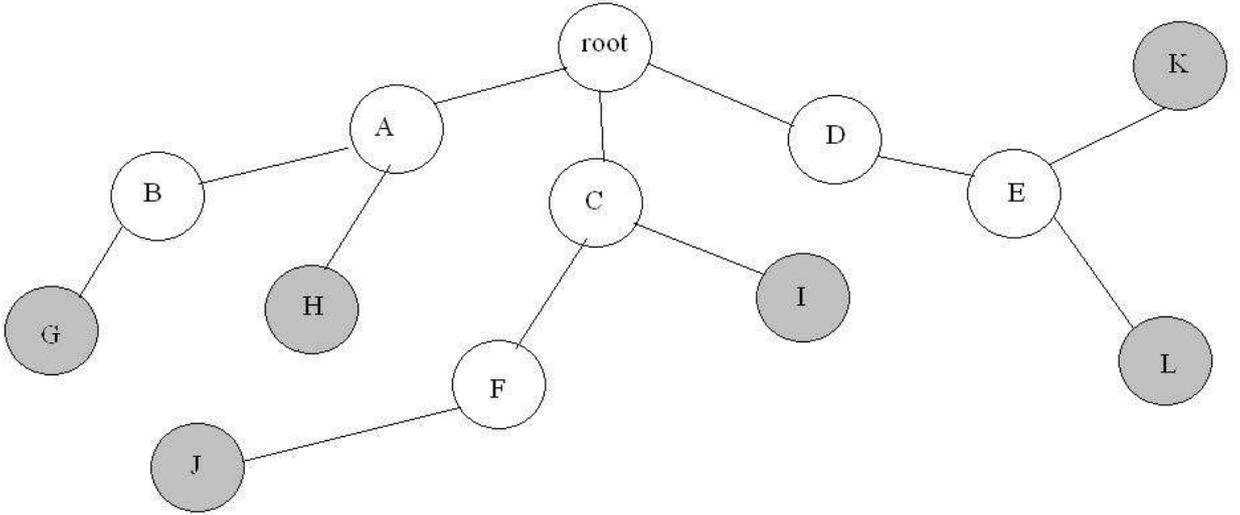}
  \caption{Example of an R-Tree}
  \label{RTree}
\end{figure}

\subsection{When is a game form Tree Representable?}

A \textit{Bad Configuration} is a combination of 2 resources $A,B$ and 3 strategies $s_1,s_2,s_3$ such that the following three conditions are satisfied:
\begin{enumerate}
\item $A,B \in s_1$
\item $A \in s_2 \setminus s_3$
\item $B \in s_3 \setminus s_2$
\end{enumerate}

%A Venn diagram describing a bad configuration is shown in Figure \ref{BadConf}
%\begin{figure}
%  \centering
%    \includegraphics[height=60mm]{BCVenn.eps}
%  \caption{A bad configuration description as a Venn diagram}
%  \label{BadConf}
%\end{figure}

Note that in subset-free congestion games there must be an additional resource $z \in s_2 \setminus s_1$, as $s_2$ is not a subset of $s_1$. Similarly, there exists a resource $w \in s_3 \setminus s_1$ (possibly $w=z$).
Verifying whether a game form is tree representable is possible in polynomial time due to the following result of Holzman and Law Yone (1996) \cite{Holz}:

\begin{theorem}\label{BadConfThm}
A congestion game form does not have a bad configuration if and only if it is tree representable.
\end{theorem}

Thus, verifying whether a game form is tree representable can be done by going over all possible pairs of resources and triplets of strategies, which is polynomial in both factors. Note that if the game is tree representable, each strategy must have at least one unique resource (the last resource on the tree path), implying that the number of resources is equal to or larger than the number of strategies.

\section{Main Result}

%\begin{definition}
%Let $F$ be a congestion game form. $F$ is a Greedy Nash Equivalent (GNE) if for any Congestion Game  $G \in {\cal G}(F)$, we have that $Z(G) = NE(G)$
%\end{definition}

%Note that from the definition of $\cal{G}$ we can concentrate on subset free congestion game forms.

Our main result links tree representable congestion games with the equivalence of the two solution concepts based on  greediness and equilibrium:

\begin{theorem}\label{main theorem 1}
Let $F$ be a subset free Congestion Game Form. $F$ is tree representable iff for any congestion game  $G \in {\cal G}(F)$, we have that $Z(G) = NE(G)$.
\end{theorem}

Recall that the set of strategies that survives deletion of dominated strategies in single-signed monotone congestion games is subset-free. Thus, the following conclusion immediately follows from Theorem \ref{main theorem 1}:

\begin{corollary}\label{main theorem 1 for single signed}
Let $F$ be a Congestion Game Form. $F$ is tree representable iff for any congestion game  $G \in {\hat{\cal G}}(F)$, we have that $Z(G) = NE(G)$,
\end{corollary}

where ${\hat{\cal G}}(F)$ denotes the set of single signed monotone congestion games with the game form $F$.

We split the proof of Theorem \ref{main theorem 1} into two propositions, one showing that tree representability is sufficient for the desired equivalence and the other showing it is necessary.

\subsection{Tree Representability is Sufficient}

As noted in \cite{holz03} TRCG are equivalent to extension parallel network games. For such games Fotakis \cite{fotakis10}, showed in Theorem 1 the following:
\begin{theorem}\label{fotakis thm}
For any n player symmetric congestion game on an extension parallel network, every best response sequence reaches a pure NE in at most n steps.
\end{theorem}

From here it is easy to show the following lemma:

\begin{lemma}\label{fotakis_addon}
Let $G$ be a TRCG. Then $Z(G) \subseteq NE(G)$.
\end{lemma}

\begin{proof}
Let $L$ be the lowest resource payoff which can be attained in $G$. Let us add a resource $r_l$ to the game $G$ with the resource payoffs $L-1,L-2, \ldots L-N$. Let us add to the set of strategies of $G$ the strategy $r_l$. Let us denote the extended game as $G'$. Let $s$ be a strategy profile of $G'$ where all players select the resource $r_l$.

Let $z$ be a greedy behavior strategy profile of $G$ with the ordering $\pi$ and a draw breaking rule $\tau$.

Let us select the first player according to the ordering $\pi$ and denote her $i_1$. Let us relocate her to the best responses to $s^{-i_1}$, in case of several the strategy selected by $\tau$. Let us denote the obtained profile as $s_1$. Similarly, let us select the second agent in $\pi$ and relocate her to the best response to $s_{1}^{-i_2}$, in case of several the strategy selected by $\tau$. After continuing in such manner once for every agent we will obtain the strategy profile $s^N$. Following theorem \ref{fotakis thm}, $s^N$ is a NE of $G'$.

Note that the best response of any agent is a strategy in $G$, as the payoff from staying on $r_l$ is strictly lower than relocating to any strategy in $G$, no matter the congestion on all other resources. Therefore, this mechanism is identical to greedy behavior of the original game, and since after $N$ steps no player selects $r_l$ we obtain a strategy profile of $G$. Note that any NE of $G'$ is also a NE of $G$, as $G'$ has one strategy more than $G$, and besides the two games are identical. Therefore, $s^N$ is a NE of $G$.

Therefore, any greedy strategy profile with ordering $\pi$ and draw breaking rule $\tau$ is a NE of the game $G$.

\end{proof}

%//*********************************************************************

\begin{lemma} \label{AllNEGreedy}

If $F$ is tree representable then in any $G \in {\cal G}(F)$, $NE(G) \subset Z(G)$.
\end{lemma}

\begin{proof}
The claim is trivial for any single player game. Let us assume it holds for $n-1$ and show that it must also hold for $n$ players.
Let $ s=( s^1,\dots s^n) \in NE(G)$. We will assume, without loss of generality, that utility of agent $n$ is the lowest among all agents:
\begin{equation}
U^n( s) \le U^i(s) \ \forall i.
\end{equation}
We denote by $P(s)$ the projection of the strategy tuple $s$ onto players $1,\ldots,n-1$.
We argue that $P( s)$ is a Nash equilibrium of the game with $n-1$ agents.

Assume it is not, then there exists an agent $j$ who can profitably deviate to some strategy $\sigma$.
Let $\tilde P(s) = (P(s)^{-j},\sigma)$ denote the strategy tuple of the $n-1$ players after such deviation. Formally, $\tilde s = (s ^{-j}, \sigma)$, and we can say that:
\begin{equation}\label{eq101}
U^j(P(\tilde s)) > U^j(P(s)).
\end{equation}

Let us denote by $\bar s = ( s^{-n},\sigma)$, the strategy profile obtained from $s$ by replacing the strategy of agent $n$ with $\sigma$. Similarly, denote  by $\ddot{s}$ the strategy tuple obtained from $s$ by replacing agent $j$'s strategy with $s^n$, namely: $\ddot s = ( s^{-j},s^n)$. As $s$ is a Nash equilibrium:
\begin{equation}\label{eq15}
U^n(s) \geq U^n(\bar s), \ \ \
U^j(s) \geq U^j(\ddot{s}).
\end{equation}

%To sum up the strategy profiles we have the following table:
%\begin{center} \begin{tabular}{c c c}
%profile & agent $j$ str. & agent $n$ str. \\
%$\hat s $ & $s^j$ & $s^n$ \\
%$\tilde s $ & $\sigma$ & $s^n$ \\
%$\bar s $ & $s^j$ & $\sigma$ \\
%$\ddot{s}$ & $s^n$ & $s^n$ \\
%\end{tabular} \end{center}

%As a first step we show that in $P(\hat s)$ a deviation from $s^j$ to $s^n$ is not a profitable one. Formally, $U^j(P(\ddot{s})) \leq U^j(P(\hat s))$.

Note the following connection between the corresponding congestion vectors:
\begin{itemize}
\item
 $\forall r \in s^j \cap s^n$,  $C(s)_r -1 = C(P(\ddot s))_r = C(P(s))_r $;
\item
  $\forall r \in s^j \setminus s^n$,  $C(s)_r = C(P(s))_r $; and
\item
$\forall r \in s^n \setminus s^j$, $C(P(\ddot s))_r = C(s)_r$.
\end{itemize}
Therefore:
\begin{eqnarray*}
U^j(P(s)) = \sum_{r \in s_j \cap s_n}P_r(C(s)_r -1) + \sum_{r \in s_j \setminus s_n}P_r(C(s)_r) \\
U^j(P(\ddot{s})) = \sum_{r \in s_j \cap s_n}P_r(C(s)_r -1) +  \sum_{r \in s_n \setminus s_j}P_r(C(s)_r)
\end{eqnarray*}
which implies
\begin{equation}\label{eq100}
U^j(P(\ddot{s})) - U^j(P(s)) = \sum_{r \in s_n \setminus s_j}P_r(C(s )_r) - \sum_{r \in s_j \setminus s_n}P_r(C(s )_r).
\end{equation}

%However, $U^j(\ddot{s}) - U^j(\hat s) = \sum_{r \in s_n \setminus s_j}P_r(C(\hat s )_r) - \sum_{r \in s_j \setminus s_n}P_r(C(\hat s )_r)$, which combined with equation \ref{eq100} implies
%$U^j(P(\ddot{s})) - U^j(P(\hat s)) = U^j(\ddot{s}) - U^j(\hat s)$. Recalling inequality \ref{eq15} we conclude that:
Since we know that $U^n(s) \leq U^j(s)$, the difference denoted in equation \ref{eq100} is non negative. Thus, we can say that:
\begin{equation}\label{allNEgrEQ11}
U^j(P(\ddot{s})) \leq U^j(P(s)).
\end{equation}

Note that the three strategies, $s^j,s^n, \sigma$ are not all identical, as $s^j$ and $\sigma$ are different.

The game $G$ is tree representable and therefore the following 3 cases, depicted in figure \ref{3casesfig}, exhaust all the possibilities on the connection among the three strategies $s^j, s^n, \sigma$ which differ on the strategy that branches off the tree path first (they are not necessarily mutually exclusive):

\begin{itemize}
\item Case 1 - $\sigma \cap s^j \cap s^n = \sigma \cap s^j = \sigma \cap s^n$ (the path representing $\sigma$ branches off first).
\item Case 2 - $\sigma \cap s^j \cap s^n = s^j \cap \sigma  = s^j \cap s^n$ (the path representing $s^j $ branches off first).
\item Case 3 - $\sigma \cap s^j \cap s^n = s^n \cap \sigma  = s^n \cap s^j$ (the path representing $s^n $ branches off first).
\end{itemize}

\begin{figure}
  \centering
    \includegraphics[height=60mm]{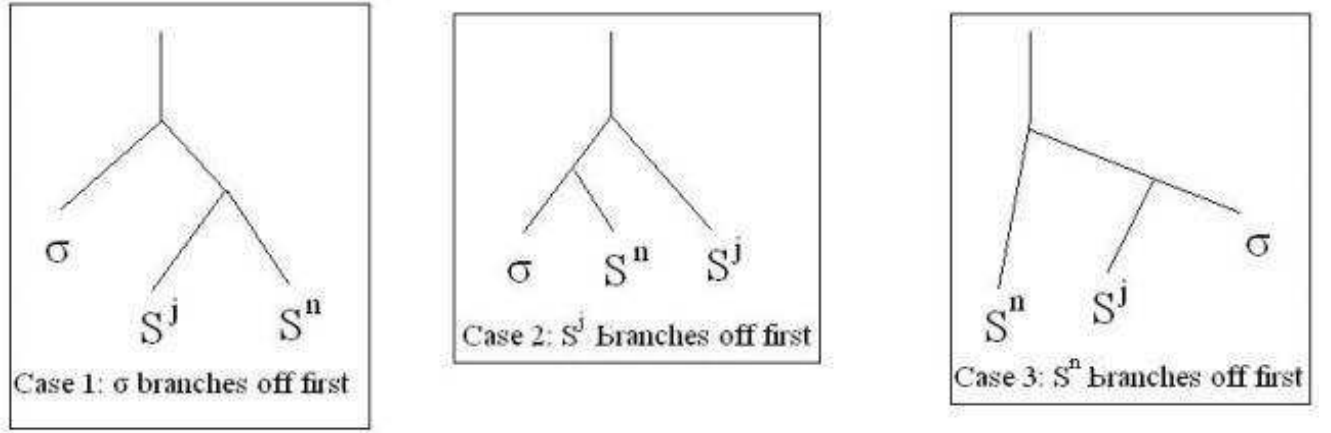}
  \caption{The three cases as described by the lemma}
  \label{3casesfig}
\end{figure}

{\bf Case 1 - $\sigma \cap s^j \cap s^n = \sigma \cap s^j = \sigma \cap s^n$}.

Since $s$ is a NE, we have that $U^j(s) \geq U^j(\tilde s)$. Therefore:
\begin{equation*}
\sum_{r \in s^j \setminus \sigma}P_r(C(s)_r)+
\sum_{r \in s^j \cap s^n \cap \sigma}P_r(C(s)_r) -
\sum_{r \in s_j \cap s^n \cap \sigma}P_r(C(\tilde s)_r) -
\sum_{r \in \sigma \setminus s_j}P_r(C(\tilde s)_r)\geq 0.
\end{equation*}

Note that $\forall r \in s^j \cap s^n \cap \sigma$,  $C(s)_r = C(\tilde s)_r$. Therefore, the second and third element cancel out. In addition, $\forall r \in s^j \setminus \sigma,$  $C(P(s))_r \leq C(s)_r$ and $\forall r \in  \sigma \setminus s^j,$  $C(P(\tilde s)_r) = C(\tilde s)_r$. As $P_r$ are strictly decreasing, the last inequality implies:
\begin{equation*}
\sum_{r \in s^j \setminus \sigma}P_r(C(P(s))_r)-
 \sum_{r \in \sigma \setminus s^j}P_r(C(P(\tilde s))_r)\geq 0.
\end{equation*}

Moreover, as $\ C(P(s))_r = C(P(\tilde s))_r \ \ \forall r \in s^j \cap s^n \cap \sigma$:
\begin{equation*}
\sum_{r \in s^j \setminus \sigma}P_r(C(P(s))_r)+
\sum_{r \in s^j \cap \sigma}P_r(C(P(s))_r)-
 \sum_{r \in \sigma \setminus s^j}P_r(C(P(\tilde s))_r)-
 \sum_{r \in s^j \cap \sigma}P_r(C(P(\tilde s))_r)\geq 0.
\end{equation*}
The last inequality implies that $U^j(P(s)) \geq U^j(P(\tilde s))$ which contradicts inequality \ref{eq101}.

{\bf Case 2 - $\sigma \cap s^j \cap s^n = s^j \cap \sigma  = s^j \cap s^n$}.

We know that $s$ is a NE, thus $U^n(s) \geq U^n(\bar s)$. Thus:
\begin{eqnarray*}
\sum_{r \in s^n \setminus \sigma}P_r(C(s)_r) +
\sum_{r \in \sigma \cap s^n} P_r(C(s)_r) \geq
\sum_{r \in \sigma \setminus s^n}P_r(C(\bar s)_r) -
\sum_{r \in \sigma \cap s^n}P_r(C(\bar s)_r) \\
\end{eqnarray*}

Note that $\forall r \in s^n \cap \sigma$,$C(\bar s)_r= C(s)_r$. Therefore the second and fourth element cancel out:
\begin{equation*}
\sum_{r \in s^n \setminus \sigma}P_r(C(s)_r) -
\sum_{r \in \sigma \setminus s^n}P_r(C(\bar s)_r) \geq 0.
\end{equation*}

Additionally, $\forall r \in \sigma \setminus s^n$,  $C(\bar s)_r= C(P(\tilde s))_r$. Similarly, $\forall r \in s^n \setminus \sigma$, $C(s)_r= C(P(\ddot s))_r$ Thus:

\begin{equation*}
\sum_{r \in s^n \setminus \sigma}P_r(C(P(\ddot{s})_r)
 - \sum_{r \in \sigma \setminus s^n}P_r(C(P(\tilde s))_r) \geq 0.
\end{equation*}

Moreover, $\forall r \in s^n \cap \sigma$, $C(P(\tilde s))_r= C(P(\ddot s))_r$ Thus:

\begin{equation*}
\sum_{r \in s_n \setminus \sigma}P_r(C(P(\ddot{s})_r) +
\sum_{r \in s_n \cap \sigma}P_r(C(P(\ddot{s})_r)-
\sum_{r \in \sigma \setminus s^n}P_r(C(P(\tilde s))_r) -
\sum_{r \in s^n \cap \sigma}P_r(C(P(\tilde s))_r) \geq 0.
\end{equation*}

We get that $U^j(P(\ddot{s})) \geq U^j(P(\tilde s))$. Combined with inequality \ref{allNEgrEQ11} we reach a contradiction with inequality \ref{eq101}.

{\bf Case 3 - $\sigma \cap s^j \cap s^n = s^n \cap \sigma  = s^n \cap s^j$}.

As $s$ is a NE, we know that $U^j(s) \geq U^j(\tilde s)$.
Therefore:
\begin{equation*}
\sum_{r \in s^j \setminus \sigma}P_r(C(s)_r)+
\sum_{r \in (s^j \cap \sigma)}P_r(C(s)_r) \geq
\sum_{r \in \sigma \setminus s^j}P_r(C(\tilde s)_r)+
\sum_{r \in (s^j \cap \sigma) }P_r(C(\tilde s)_r).
\end{equation*}

Note that $\forall r \in s^j \cap \sigma $, $C(\tilde s)_r= C(s)_r$, thus the second and fourth element cancel out. Additionally, $\forall r \in s^j \cap \sigma $, $C(P(\tilde s))_r= C(P(s))_r$ Therefore:
\begin{equation*}
\sum_{r \in s^j \setminus \sigma}P_r(C(s)_r)+
\sum_{r \in s^j \cap \sigma}P_r(C(P(s))_r)\geq
\sum_{r \in \sigma \setminus s^j}P_r(C(\tilde s)_r)+
\sum_{r \in s^j \cap \sigma}P_r(C(P(\tilde s))_r).
\end{equation*}

Moreover, $\forall r \in (s^j \setminus \sigma)$, $C(P(s))_r=C(s)_r$ and $\forall r \in (\sigma \setminus s^j)$, $C(\tilde s)_r= C(P(\tilde s))_r$. Therefore:

\begin{equation*}
\sum_{r \in s^j \setminus \sigma}P_r(C(P(s))_r)+
\sum_{r \in s^j \cap \sigma}P_r(C(P(s))_r)\geq
\sum_{r \in \sigma \setminus s^j}P_r(C(P(\tilde s))_r)+
\sum_{r \in s^j \cap \sigma}P_r(C(P(\tilde s))_r).
\end{equation*}

The last inequality states that $U^j(P(s)) \geq U^j(P(\tilde s))$, once again contradicting inequality \ref{eq101}.

Thus, $P(s)$ must be a Nash equilibrium for the $n-1$ players. Using the induction hypothesis we conclude that $P(s)$ is a greedy profile for players $1,\ldots,n-1$. As player $n$ best-replies to $P(s)$ (recall that $s$ is a Nash equilibrium) we conclude that $s$ is a greedy profile, as desired.
\end{proof}

Combining lemmas \ref{fotakis_addon} and \ref{AllNEGreedy} yields the following direction in the statement of Theorem \ref{main theorem 1}:
\begin{proposition} \label{TRGG_GNE}
Let $F$ be a subset free congestion game form. If $F$ is tree representable then $G \in {\cal G}(F) \implies Z(G) = NE(G)$.
\end{proposition}

\subsection{Tree Representability is Necessary}

We show that if a game form is non tree representable then it can be coupled with monotone resource payoff function to yields a game, $G$, without the equivalence of $NE(G)$ and $Z(G)$. Let $F=(R,\Sigma)$ be a game form that is not tree representable. By Theorem \ref{BadConfThm} there must exist 2 resources, $A,C \in R$ and three strategies, $s_1,s_2,s_3 \in \Sigma$ such that $A \in (s^1 \cap s^3) \setminus s^2$ and $C \in (s^1 \cap s^2) \setminus s^3$.

Recall that $F$ is assumed subset free. In particular there must exist resources $B$ and $D$ that satisfy
$B \in s^2 \setminus s^1$ and $D \in s^3 \setminus s^1$. We now argue that $B \not = D$ and more broadly that:

%\begin{figure}
% \centering
%    \includegraphics[height=60mm]{NTProof.eps}
%  \caption{Strategies and Resources in a Bad Configuration}
%  \label{NTProof}
%\end{figure}

\begin{lemma}\label{no_resources}
If for any $G \in {\cal G}(F)$, $Z(G)=NE(G)$ then $(s^2\cap s^3) \setminus s^1 = \emptyset$.
\end{lemma}

\begin{proof}
Suppose that there exists some resource $E \in (s^2\cap s^3) \setminus s^1$. Then we have that:
\begin{eqnarray*}
s^1 \cap \{A,C,E\} =\{A,C\}\\
s^2 \cap \{A,C,E\} =\{C,E\}\\
s^3 \cap \{A,C,E\} =\{A,E\}
%A \in (s^1 \cap s^3)\setminus s^2\\
%C \in (s^1 \cap s^2)\setminus s^3\\
%E \in (s^3 \cap s^2)\setminus s^1
\end{eqnarray*}

Consider a 2 player game with the following payoff functions:

\begin{tabular}{|c|c|c|c|c|}
\hline
\# of players / Resource & A & C & E & other resources\\
\hline
1 & 10 & 9 & 8 & $\frac{1}{M}$\\
\hline
2 & 1 & 6 & 7 & ${\frac{1}{2M}}$\\
\hline
\end{tabular}\\

Let $M$ be sufficiently large to ensure that $\frac{2|R|}{M} <1$. This game does not satisfy $Z(G)=NE(G)$: for example, the NE involving AE and EC cannot be attained in a greedy manner. Thus we reach a contradiction.
\end{proof}

\begin{corollary}
$B$ and $D$ are two different resources.
\end{corollary}

\begin{lemma}\label{exists s4}
If for any $G \in {\cal G}(F)$, $Z(G)=NE(G)$ then  there exists a strategy $s^4 \neq s^2$ such that $s^4 \subset s^2 \cup (s^3 \setminus s^1)$.
\end{lemma}

\begin{proof}
By Lemma \ref{no_resources} $(s^2 \cap s^3) \setminus s^1 = \emptyset$. Therefore the following assignment of resource payment functions for 2 players determines a monotone congestion game (where $M$ is an arbitrary large number satisfying $M>R^9$):

\begin{tabular}{c c c}
Resource Set & $P_r(1)$ & $P_r(2)$ \\
$s^1 \cap s^2 \cap s^3$ & $-1/M^2$ & $-1/R$ \\
$(s^1 \cap s^2) \setminus s^3$ & $-1/M^2$ & $-|R|^6$ \\
$(s^1 \cap s^3) \setminus s^2$ & $-1/|R|^5$ & $-2M$ \\
$s^1 \setminus (s_2 \cup s^3)$ & $-1/|R|^5$ & $-2M$ \\
$s^2 \setminus (s^1 \cup s^3)$ & $-1/|R|$ & $-2M$ \\
$s^3 \setminus (s^2 \cup s^1)$ & $-1/|R|^4$ & $-2M$ \\
$(s^1 \cup s^2 \cup s^3)^c$ & $-M$ & $-2M$
\end{tabular}

Note that resources not in $s^1$ are worse than any resource in $s^1$ by a factor of $|R|$. Therefore the first greedy player  must select the strategy $s^1$. The second greedy agent will need to select resources in $s^2 \cup (s^3\setminus s^1)$, as all other resources have utility of at least -M. Assume the statement of the lemma is incorrect and that the only such strategy is $s^2$ and so player 2 chose strategy $s^2$. Now note that the first greedy agent has a profitable deviation from $s^1$ to $s^3$, avoiding the high negative payoff on resources in $(s^1 \cap s^2) \setminus s^3$ (e.g., on resource $C$). This implies that $Z(G) \not = NE(G)$ and a contradiction is reached.
\end{proof}

%\begin{remark}
%The following holds for any number of agents in $G'$. Simply let the first $|N|-2$ agents select $s_1$ (or $s_3$), by having there the costs bove $\varepsilon / |R|$ for some sufficiently small $\varepsilon$.
%\end{remark}

Combined with subset-freeness of the strategy set we can now conclude:

\begin{corollary}\label{4strCor}
There exists a resource in $s^4$ which is also in $s^3 \setminus (s^1 \cup s^2)$. We can assume WLOG that this resource is $D$.
\end{corollary}

\begin{proposition}\label{sufficient}
Let $F$ be a subset free congestion game form. Then $Z(G)=NE(G)\ \  \forall G \in {\cal G}(F)$ implies $F$ is tree representable.
\end{proposition}

\begin{proof}
Suppose $F$ is not tree representable and $Z(G)=NE(G)\ \  \forall G \in {\cal G}(F)$. Then, by Lemmas \ref{no_resources} and \ref{exists s4}, there exist 4 strategies $s^1,\ldots,s^4$ and 4 distinct resources, $A,B,C,D$ satisfying:
$A \in (s^1 \cap s^3) \setminus s^2$, $B \in s^2 \setminus s^1$, $C \in (s^1 \cap s^2) \setminus s^3$ and
$D \in s^3 \setminus (s^1 \cup s^2)$. We will now show that in such case $s^4 \subset s^3$, thus contradicting the subset-freeness assumption.

Suppose there exists a resource $E \in s^4 \setminus s^3$. This in turn implies that $E \in s^2 \setminus s^3$. Let us distinguish between the case where $E \in s^1$ and the case $E \not \in s^1$.

{\bf Case 1 - Assume $E \in s^1$}. This, in particular implies that:
\begin{eqnarray*}
s^1 \cap \{A,D,E\} =\{A,E\}\\
s^3 \cap \{A,D,E\} =\{A,D\}\\
s^4 \cap \{A,D,E\} =\{D,E\}
%A \in (s^1 \cap s^3)\setminus s^4\\
%D \in (s^3 \cap s^4)\setminus s^1\\
%E \in (s^1 \cap s^4)\setminus s^3
\end{eqnarray*}

Consider a 2 player game with the following payoff functions:

\begin{tabular}{|c|c|c|c|c|}
\hline
\# of players / Resource & A & D & E & other resources\\
\hline
1 & 10 & 9 & 8 & $ \frac{1}{M}$\\
\hline
2 & 1 & 6 & 7 & $\frac{1}{2M}$\\
\hline
\end{tabular}\\

Let $M$ be sufficiently large to ensure that $\frac{2|R|}{M} <1$. This game does not satisfy $Z(G)=NE(G)$, as for example the NE involving AE and DE cannot be attained by greedy behavior. Thus we reach a contradiction.

{\bf Case 2 - Assume $E \not \in s^1$}. This, in particular implies that:
\begin{eqnarray*}
s^1 \cap \{A,C,D,E\} =\{A,C\}\\
s^2 \cap \{A,C,D,E\} =\{E,C\}\\
s^3 \cap \{A,C,D,E\} =\{A,D\}\\
s^4 \cap \{A,C,D,E\} =\{E,D\}
%A \in (s^1 \cap s^3)\setminus (s^2 \cup s^4)\\
%C \in (s^1 \cap s^2)\setminus (s^3 \cup s^4)\\
%D \in (s^3 \cap s^4)\setminus (s^1 \cup s^2)\\
%E \in (s^2 \cap s^4)\setminus (s^1 \cup s^3)
\end{eqnarray*}

Consider the 2 player game with the following resource payment functions:

\begin{tabular}{|c|c|c|c|c|c|}
\hline
\# of players / Resource & A & E & C & D & other resources\\
\hline
1  & 40 & 30 & 20 & 15 & $\frac{1}{M}$\\
\hline
2  & 10 & 11 & 12 & 13 &  $\frac{1}{2M}$ \\
\hline
\end{tabular}\\

Let $M$ be sufficiently large to ensure that $\frac{2|R|}{M} <1$. This game is in the spirit Example \ref{2layer} (with the addition of extra resources that yield negligible utility). Similar to Example \ref{2layer} this game does not satisfy $Z(G)=NE(G)$, thus reaching a contradiction.
\end{proof}

The proof of Theorem \ref{main theorem 1} follows from Propositions \ref{TRGG_GNE} and \ref{sufficient}.

\section{Summary}

Monotone congestion games are a well proved modeling tool. In many realistic cases one can easily assume that the strategy set is subset-free (e.g., when strategies are paths leading from a source node to a target node) or that resource payment functions are single-signed (e.g., they express latency over a graph edge and are hence negative). We show that in such cases the set of pure Nash equilibria coincides with the set of greedy strategy profiles. We conclude that in such cases a Nash equilibrium forms a viable solution concept as it emerges from very weak rationality assumptions and does not hinge on common knowledge of rationality. In addition, it can be the case that the computational difficulty of finding the equilibria set in such games is substantially weaker than in an arbitrary game or even an arbitrary congestion game.


\begin{thebibliography}{label}

\bibitem{Ackerman et al} Ackermann H, H. Roglin, and  B. Vocking, 2006, On the Impact of
Combinatorial Structure on Congestion Games, proceedings of the 47th IEEE Symp. on foundation of Computer Science (FOCS06), 613-622
on Computational Complexity, Report No. 67.

\bibitem{Fabrikant et al} Fabrikant A, C. Papadimitriou and K. Talwar, 2004, The
complexity of pure Nash equilibria, The 36th Annnual Association for
Computing Machinery Symposium on Theory of Computation, 604--612.

\bibitem{Fotakis05} Fotakis D, S. Kontogiannis and
P. Spiraklis, 2005, Symmetry in Network Congestion Games: Pure
equilibria and Anarchy Cost, Workshop on Approximation and Online
Algorithms 2005, Lecture Notes in Computer Science 3879, 161--175.

\bibitem{fotakis10} Fotakis D., 2010, Congestion games with linearly independant paths: convergence time and price of Anarchy, Theory of Computing Systems 47, 113-136

\bibitem{Holz} Holzman R, N. Law Yone, 1997, Strong Equilibrium in Congestion Games, Games and Economic Behavior 21, 85-101

\bibitem{holz03} Holzman R, N. Law Yone, 2003, Network Structure and strong equilibrium in route selection games, Mathematical social sciences 46, 193-205

\bibitem{Ieong} Ieong S., R. McGrew, E. Nudelman, Y. Shoham, Q. Sun, 2005, Fast and Compact: A Simple class of congestion games, AAAI.

\bibitem{milchteich} Milchtaich I., 2006, Network Topology and the efficiency of equiliblria, Games and Economic Behavior 57, 321-346

\bibitem{MondSh} Monderer D., L. S. Shapley, 1996, Potential Games, Games and Economic Behavior 14, 124-143.

\bibitem{Rosental} Rosental R.W., 1973, A class of Games Posessing Pure Strategy Nash Equilibria, International Journal of Game Theory, vol. 2, 65-67

\end{thebibliography}
\end{document}